\documentclass[conference]{IEEEtran}
\IEEEoverridecommandlockouts
 \pdfoutput=1
 
\usepackage{balance}
\usepackage{graphicx}
\usepackage{caption}
\usepackage [noadjust]{cite} %[nospace]
\usepackage{bm}  % bold math for greek symbols
\usepackage{amsmath}
\usepackage{amssymb}
\usepackage{enumerate}
\usepackage{stfloats}
\usepackage{cases}
\usepackage{epstopdf}
\usepackage{soul,color} % \ul  \hl
\usepackage{hyperref}
\usepackage[percent]{overpic}
\usepackage{tabularx}
\usepackage[table]{xcolor}
\usepackage{multirow}
\usepackage{booktabs}  % multicolumn tabular
\usepackage{tabu} % different font sizes in different rows of tabular

\usepackage{algorithm}
\usepackage{algpseudocode}

\newcommand{\pbold}{\boldsymbol{p}}

\newcommand{\dF}{d^{\mathrm{F}}}
\newcommand{\dFMax}{\overline{d}^{\mathrm{F}}}
\newcommand{\dFzero}{d^{\mathrm{F0}}}
\newcommand{\dFzeroMax}{\overline{d}^{\mathrm{F0}}}
\newcommand{\dN}{d^{\mathrm{N}}}
\newcommand{\dNMax}{\overline{d}^{\mathrm{N}}}
\newcommand{\dNzero}{d^{\mathrm{N0}}}
\newcommand{\dNzeroMax}{\overline{d}^{\mathrm{N0}}}

\newtheorem{lemma}{Lemma}

\newtheorem {remark}{Remark}
\def\BibTeX{{\rm B\kern-.05em{\sc i\kern-.025em b}\kern-.08em
    T\kern-.1667em\lower.7ex\hbox{E}\kern-.125emX}}
\begin{document}
	\raggedbottom
	\allowdisplaybreaks

     \title{Revisiting the Fraunhofer and Fresnel Boundaries for Phased Array Antennas%s in Near-Field Communication
     \thanks{This work was supported by Business Finland via the 6GBridge - Local 6G project (grant number 8002/31/2022) and by the Research Council of Finland through the 6G Flagship (grant number 346208).}}

% \author{\IEEEauthorblockN{Mehdi Monemi, Mehdi Rasti, Matti Latva-aho}\\
% \IEEEauthorblockA{\textit{Centre for Wireless Communications (CWC), University of Oulu, Oulu, Finland} \\
% Email: mehdi.monemi@oulu.fi, mehdi.rasti@oulu.fi, matti.latva-aho@oulu.fi}
% }

    %\author{Mehdi Monemi, \textit{Member}, IEEE, Mehdi Rasti, \textit{Senior Member}, IEEE, and  Matti Latva-Aho, \textit{Fellow}, IEEE
    
    %\title{A Novel Distributed DRL-based WPT Structure using Dense Metasurfaces in the Fresnel Zone}
	%\author{Mehdi Monemi, \textit{Member}, IEEE, Mohammad Amir Fallah, Mehdi Rasti, \textit{Senior Member},IEEE, Matti Latva-Aho, \textit{Senior Member}, IEEE,  and Merouane Debbah, \textit{Fellow}, IEEE
		%\thanks{\ }
	%}
    	
%, Matti Latva-Aho, \textit{Senior Member}, IEEE	

% 	\vspace{-15mm}

\author{
{Mehdi Monemi, Mehdi Rasti, and Matti Latva-aho}
\\Center for Wireless Communications (CWC),
University of Oulu,
Oulu, Finland \\
Email: \{mehdi.monemi, mehdi.rasti, matti.latva-aho\}@oulu.fi}
% \and
% \IEEEauthorblockN{Mehdi Rasti,\textit{ Senior Member, IEEE}}
% \IEEEauthorblockA{\textit{Center for Wireless Comm. (CWC)} \\
% \textit{University of Oulu}\\
% Oulu, Finland \\
% mehdi.rasti@oulu.fi}
% \and
% \IEEEauthorblockN{Matti Latva-Aho, \textit{Fellow, IEEE}}
% \IEEEauthorblockA{\textit{Center for Wireless Comm. (CWC)} \\
% \textit{University of Oulu}\\
% Oulu, Finland \\
% matti.latva-aho@oulu.fi}	
 \maketitle
\begin{abstract}
This paper presents the characterization of near-field propagation regions for phased array antennas, with a particular focus on the propagation boundaries defined by Fraunhofer and Fresnel distances. These distances, which serve as critical boundaries for understanding signal propagation behavior, have been extensively studied and characterized in the literature for single-element antennas. However, the direct application of these results to phased arrays, a common practice in the field, is argued to be invalid and non-exact. This work calls for a deeper understanding of near-field propagation to accurately characterize such boundaries around phased array antennas.
More specifically,
for a single-element antenna, the Fraunhofer distance is $\dF=2D^2 \sin^2(\theta)/\lambda$ where $D$ represents the largest dimension of the antenna, $\lambda$ is the wavelength and $\theta$ denotes the observation angle. We show that for phased arrays, $\dF$ experiences a fourfold increase (i.e., $\dF=8D^2 \sin^2(\theta)/\lambda$) provided that $|\theta-\frac{\pi}{2}|>\theta^F$ (which holds for most practical scenarios), where $\theta^F$ is a small angle whose value depends on the number of array elements, and for the case $|\theta-\frac{\pi}{2}|\leq\theta^F$, we have $\dF\in[2D^2/\lambda,8D^2\cos^2(\theta^F)/\lambda]$, where the precise value is obtained according to some square polynomial function $\widetilde{F}(\theta)$. Besides, we also prove that the Fresnel distance for phased array antennas is given by $\dN=1.75 \sqrt{{D^3}/{\lambda}}$ which is $\sqrt{8}$ times greater than the corresponding distance for a conventional single-element antenna with the same dimension.
\end{abstract}
	%...................................................................................................................................
	% keywords
	
\maketitle	%\IEEEpeerreviewmaketitle
	
	%...................................................................................................................................
	% Introduction
%\thispagestyle{empty}

\section{Introduction}
\label{sec:introduction} 
%Near-field communication has recently garnered significant interest in theoretical and practical research.
The application of near-field communication is expected to significantly expand in existing and future generations of wireless technology compared to previous generations \cite{10243590, zhang20236g}. This phenomenon is attributable to two key factors. Firstly, modern applications gravitate towards higher frequencies in the millimeter wave (mmWave) and terahertz (THz) bands \cite{10040913}. %, which correspond to very short signal wavelengths. 
Secondly, due to the severe path loss experienced at such high frequencies, these applications employ extremely large-scale antenna arrays (ELAAs) to compensate for this loss through high-directivity beamforming schemes \cite{10379539}. Noting that the near-field propagation region extends by increasing the frequency and antenna diameter, these advancements have collectively led to a substantial expansion of the near-field region for electromagnetic (EM) wave propagation in modern high-frequency applications exploiting ELAAs. This expansion has the potential to extend the reach of near-field communication to hundreds of meters in many applications, as highlighted in recent studies \cite{zhang20236g, cui2022near6G}. 
Near-field communication and propagation provide a higher degree of freedom (DoF), paving the way for innovative applications that are not possible in the far-field. For example, it facilitates beamforming in both radial and angular dimensions, which opens up the potential for 3D beamfocusing. Moreover, the near-field presents unique challenges and opportunities such as locally changing polarization of EM waves, non-stationary characteristics of antenna elements, non-sparse channel representation of ELAAs, and holographic attributes of antenna arrays. In this scenario, to effectively navigate the intricacies of near-field communication, leverage its distinctive features, and build sturdy and efficient modern wireless systems, a profound comprehension of near-field signal behavior and precise characterization of near-field boundaries is essential, especially when it originates from phased array antennas.

The signal propagation region can be divided into 3 subregions differentiated by the boundaries named {\it Fraunhofer distance}, {\it Fresnel distance}, and {\it Non-radiative distance}. The Fraunhofer distance determines the transition boundary between the near-field and far-field.
    The Fresnel distance relates to a specific phase delay approximation inside the near-field region, as explained later, and finally, the non-radiative distance is the distance below which the non-radiative reactive power dominates the active power. The latter is very close to the antenna, usually lower than half a wavelength \cite{9282256, balanis2016antenna}.
The {\it Fraunhofer} and {\it Fresnel} distances have been extensively studied and characterized in the literature for single-element antennas. These distances serve as critical boundaries in understanding antenna behavior. However, when it comes to phased arrays, a common practice has been to apply these results without modification directly. Our work rigorously argues that this approach is invalid and non-exact.
In this paper, we revisit existing definitions, perform detailed calculations, and provide characterizations for the Fraunhofer and Fresnel boundaries by considering phased array antennas, and present closed-form expressions.   %Additionally, we identify potential misconceptions present in the literature. 
Notably, we emphasize the inaccuracies arising from directly applying single-element antenna characterization principles to phased arrays, as observed in prior works. A schematic representation of the stated regions developed in existing works for a single-element antenna and those characterized in this work for phased array antennas having the same dimensions as the single-element one are depicted in Fig. \ref{fig:regions} as dotted and dashed lines respectively.

\section{Background and Contributions}
In what follows, we explore the background and contributions relating to the characterization of Fraunhofer and Fresnel distances for phased array antennas. Before starting the literature review, we present some terms and definitions used in the following parts. We define the {\it principle axis} of the antenna as the axis perpendicular to the antenna surface, passing from the geometrical central point of the antenna. The {\it boresight} is the axis corresponding to the main lobe of the antenna. This is identical to the principal axis for single-element center-fed symmetrical antennas, however, for phased arrays, we use this term to represent the axis connecting the geometrical antenna center and the intended observation/transmitting point source. The principal axis and boresight are depicted for single-element and phased array antennas in Figs. \ref{fig:Fraunhofer_single} and \ref{fig:Fraunhofer_revisit} respectively. The {\it on-boresight} and {\it off-boresight} refer to the cases where the principal axis does and does not, respectively, coincide with the boresight of the antenna.

{\hspace{-14pt} \it 1- \underline{Fraunhofer region}}

{\bf Background:} 
The Fraunhofer distance, also known as the Rayleigh distance, separates the near-field region dominated by spherical wave propagation from the far-field region where the planar wave approximation holds.
It is well-known that a maximum phase delay of $\pi/8$ between a point on the antenna surface and the feed, caused by the curvature of the wavefront, is small enough to approximate the spherical wavefront as planar. This phase delay has been accepted in the literature for defining and calculating the Fraunhofer distance. %Several works have calculated 
The maximum value of the Fraunhofer distance for a center-fed single-element antenna is obtained on the boresight as $\dFzeroMax=2D^2/\lambda$, where $D$ is the maximum dimension of the antenna and $\lambda$ is the wavelength. For the case where the observation point is not on the boresight, the Fraunhofer distance can be obtained as $\dFzeroMax=2D^2\sin^2(\theta)/\lambda$ where $\theta$ is the observation angle which is the angle between the antenna plane, and the line connecting the radiating/observing point and the antenna reference point (i.e., the feed point at the center) %on the antenna center, 
as depicted in Fig. \ref{fig:Fraunhofer_single}. %and \ref{fig:Fraunhofer_revisit}. 
Several works have calculated the Fraunhofer distance for a single-element antenna following different schemes. The most straightforward scheme is to calculate the dominant terms of the binomial or Taylor series expansion approximation of the signal, and then obtain the maximum distance from the antenna where the dominant terms would result in a {\it differential phase delay} no more than $\pi/8$ on different points of the antenna surface \cite{balanis2016antenna,stutzman2012antenna}. As an alternative solution, the authors of \cite{selvan2017fraunhofer} have derived $\dFzeroMax$ by considering a continuous aperture of arbitrary shape, and then they employed scalar diffraction theory for obtaining Fraunhofer and Fresnel regions. In \cite{10541333}, the {\it effective Fraunhofer distance} $\dFzeroMax_{\mathrm{eff}}$ is defined and characterized as the distance where the normalized beamforming gain under the far-field assumption is no less than a value $\eta$. By considering $\eta=95\%$, it has been shown that we have $\dFzeroMax_{\mathrm{eff}}=0.367\dFzeroMax$ on the boresight of the antenna.
To the best of our knowledge, all works in the literature except \cite{hu2023design} have considered the same Fraunhofer values characterized based on the single-element antennas to apply to phased arrays as well (e.g., \cite{cui2022near6G,yan2021joint,cui2022near,9978148}). The authors of \cite{hu2023design}  have partially investigated the fact that for the case of an off-boresight phased array scenario, the Fraunhofer distance is increased compared to a single-element antenna, however, a comprehensive and analytical investigation is missing. In particular, a detailed analysis of how the observation angle $\theta$  exactly impacts the Fraunhofer distance, and how it is changed softly when moving from off-boresight toward on-boresight angles is missing.  

{\bf Contributions:}
We delve into an analysis of the well-established derivation of the Fraunhofer distance, which traditionally relies on the single-element antenna model, and show that this procedure does not directly translate to phased arrays. As a result, the widely accepted values of the Fraunhofer distance for phased arrays, as found in most existing literature, are inexact. Then we derive a closed-form value for the Fraunhofer distance for phased arrays for $\theta\in[0,\pi]$. More specifically, we define and calculate the {\it Fraunhofer array angle} $\theta^{\mathrm{F}}$ and show that for $\theta\in[0,\pi/2-\theta^{\mathrm{F}}] \cup [\pi/2+\theta^{\mathrm{F}},\pi]$ the Fraunhofer distance is $\dF=8D^2\sin^2(\theta)/\lambda$ which is four times $\dFzero$ relating to  
a single-element antenna with the same diameter $D$. For the case where $\theta$ is softly changed from $\pi/2\pm \theta^{\mathrm{F}}$ toward $\pi/2$ (where $\pi/2$ corresponds to the boresight), $\dF$ softly switches to $2D^2/\lambda$ according to a function whose value is derived in a closed form.
We obtain a tight closed-form approximation for the value of $\theta^{\mathrm{F}}$ and show that it is a small angle, its value being a function of the number of array elements; as the antenna array scales up, this angle tends toward zero. The schematic presentation of the characterized Fraunhofer distance $\dF$ as well as that relating to a center-fed single-element antenna (i.e., $\dFzero$) are depicted in Fig. \ref{fig:regions} per $\theta$ as the dashed and dotted blue lines respectively.
%Phased array antennas, with their ability to steer beams electronically and provide MIMO communications, have revolutionized various fields, from radar and communications to medical imaging. Optimal resource allocation for phased array antennas requires careful consideration of the radiated field's propagation characteristics. This is where the concept of near-field and far-field regions becomes crucial. 
%In general, the signal is radiated in a spherical manner
%The near-field region, existing close to the antenna, is characterized by complex wave interactions between individual elements leading to non-planar wavefronts. In contrast, the far-field exhibits a plane wave behavior, simplifying analysis and predictions. Accurately defining and understanding these propagation distances are essential for various applications. For instance, characterizing antenna properties like gain and pattern accurately requires measurements within the far-field.

%The increasing prominence of large-scale phased arrays and intelligent surfaces further magnifies the importance of near-field understanding. These advanced systems often operate at higher frequencies for increased bandwidth and resolution, further expanding the near-field region. Analyzing and characterizing this complex near-field becomes critical for optimizing their performance and mitigating potential interference issues. 
{\it \hspace{-14pt} 2- \underline{Fresnel distance}}

{\bf Background:}  The Fresnel distance is presented as the distance from the transmit antenna in the near-field region where the curvature of the wavefront is mainly due to the second term of binomial expansion of $r'$ in terms of $r$, and the maximum phase delay relating to the third term is limited to $\pi/8$. Several works have calculated the maximum Fresnel distance based on the single-element antenna using scalar diffraction theory or directly employing the binomial expansion or Taylor series \cite{stutzman2012antenna,balanis2016antenna,selvan2017fraunhofer}. Several slightly different values of the Fresnel distance have been presented in the literature. The most well-known Fresnel distance function is presented for a center-fed single-element antenna as $d^{\mathrm{N0}}= \sqrt{ |\cos \theta|\sin^2 \theta  \frac{D^3}{\lambda}}$, which corresponds to a maximum Fresnel distance of $\dNzeroMax=0.62\sqrt{{D^3}/{\lambda}}$ \cite{balanis2016antenna}.

{\bf Contributions:}
Similar to the discussion relating to the Fraunhofer region, we justify that the simple model of a center-fed single-element antenna does not apply to calculate the exact value of the Fresnel distance for phased arrays, and thus we need to revisit the existing relations. We derive a multi-relation expression to represent the Fresnel distance as a function of the observation angle $\theta$ and array size. Additionally, we prove that the  Fresnel distance for phased arrays is given by $ \dNMax=1.75 \sqrt{{D^3}/{\lambda}}$ which is $\sqrt{8}$ times greater than that of $\dNzeroMax$ associated with a single-element antenna having the same diameter as the phased array.

\begin{figure*}
    \centering
 \includegraphics[width=384pt]{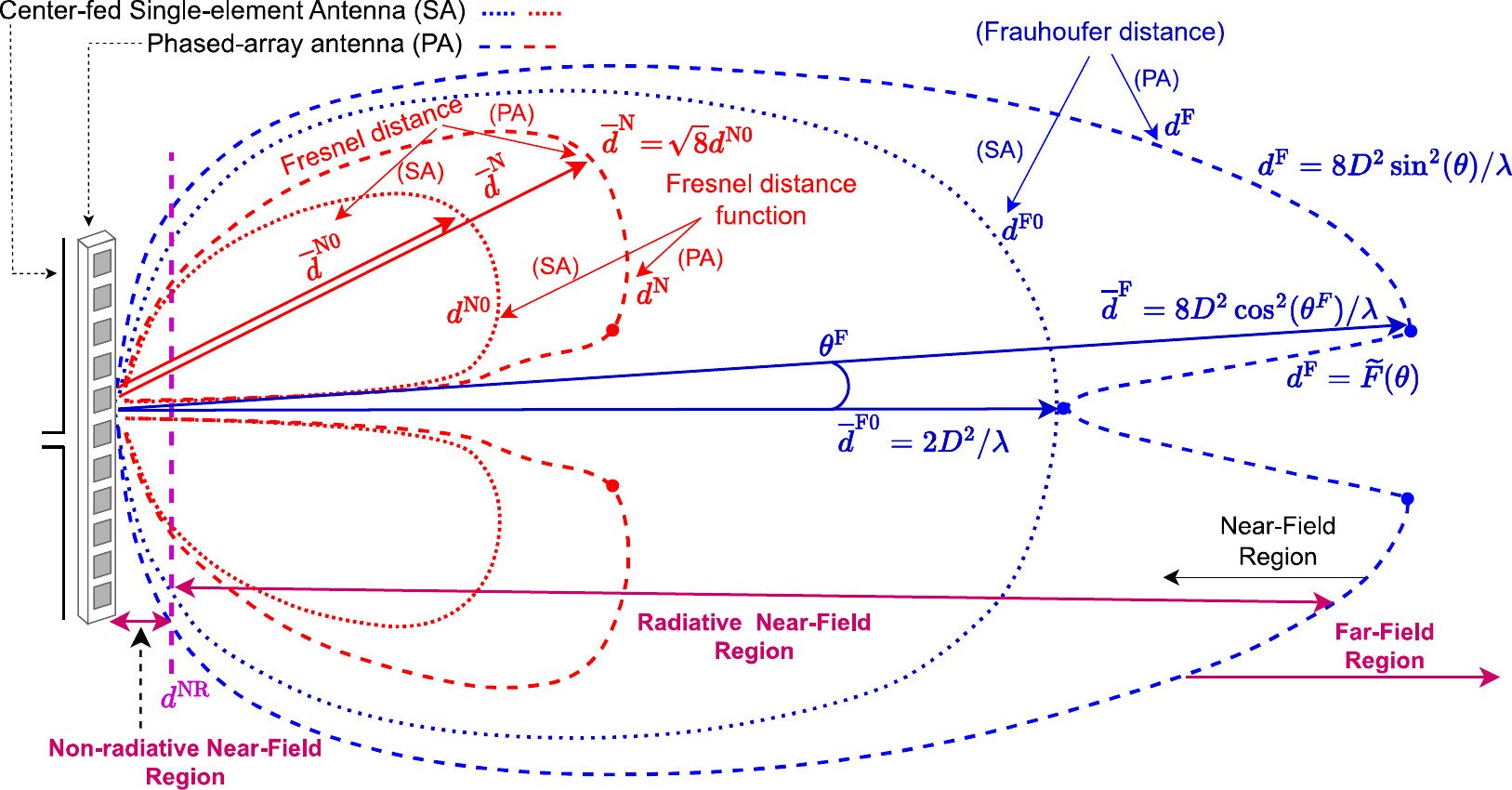}
%\vspace{-15pt}    
\caption{Franhofer and Fresnel characterized regions for a phased array antenna (dashed lines), versus a center-fed single-element antenna having the same dimension (dotted lines).} 

    \label{fig:regions}
    % \vspace{-10mm}
\end{figure*}

\section{Characterization of Fraunhofer and Fresnel Regions}
\label{sec:Fraunhofer_Fresnel}
In what follows, first, we review the definitions and calculations of the Fraunhofer and Fresnel distances characterized for a single-element antenna \cite{balanis2016antenna,selvan2017fraunhofer}. Then, we explain the limitations of applying these values to phased array antennas and propose a new approach to determine these values for such antennas.

\subsection{Fraunhofer and Fresnel distance based on the single-element antenna model}
First, consider a transmitting point source and a single-element center-fed receiving antenna with diameter $D$ as depicted in Fig. \ref{fig:Fraunhofer_single}-a. Let $r$ and $r'$ be the distance from the transmitting source to the center  (feed point) and some point $\pbold'$ on the antenna respectively, where the distance between the antenna center and point $\pbold'$ is denoted by $d'$, and the angle between the lines connecting respectively the transmitting point and $\pbold'$ to the feed is denoted by $\theta$. We can write $r'$  in terms of $r$,  $d'$, and $\theta$ as follows:
\begin{align}
    r'=%\sqrt{(x-x')^2+(y-y')^2+(z-z')^2}\big|_{(x',y',z')=(0,0,z')}
    %\\
    \sqrt{r^2+(-2rd'\cos \theta +(d')^2) }
\end{align}
By using the binomial expansion, the phase difference between the signals arrived at $\pbold'$ and the antenna center is obtained as 
% \begin{multline}
%     \Delta\theta=\frac{2\pi}{\lambda}\left(r'-r\right)=
%     \\
%     \underbrace{-\frac{\pi}{\lambda}D\cos \theta}_{\Delta\theta_1}
%     +\underbrace{\frac{\pi}{4\lambda} D^2 r^{-1}\sin^2 \theta}_{\Delta\theta_2^N}
%     +\underbrace{\frac{\pi}{8\lambda} D^3 r^{-2}\cos \theta\sin^2 \theta}_{\Delta\theta_3}
%     + ...
% \end{multline}
\begin{multline}
\label{eq:deltatheta}
    \Delta\theta=\frac{2\pi}{\lambda}\left(r'-r\right)
=
    \underbrace{\frac{-2\pi}{\lambda}\cos \theta d' }_{\Delta\theta_1}
+
\underbrace{\frac{\pi}{\lambda}  \sin^2 \theta  (d')^2 r^{-1}}_{\Delta\theta_2}
+
\\
\underbrace{\frac{\pi}{\lambda} \cos \theta\sin^2 \theta (d')^3 r^{-2}}_{\Delta\theta_3}
    + \ ...
\end{multline}
It is seen that the first term $\Delta \theta_1$ is independent of the distance $r$. %Thus, we start with the second term.%The Fraunhofer distance is defined as the distance where a maximum phase error experienced by $\Delta \phi_2$ equals $\pi/8$.
%It is well known that for most practical antennas, with overall lengths greater than a wavelength, a maximum total phase error of $\pi/8$ is not very detrimental in the analytical formulations. 
It is well known that a phase error equal to $\pi/8$  caused by the curvature of the wavefront is small enough to approximate the spherical wavefront as planar. Therefore, the Fraunhofer distance denoted by $\dFzero$ is defined as the boundary limit for which the maximum phase error between some point on the antenna and the feed is equal to $\pi/8$. i.e.,
\begin{align}
\label{eq:Fraunhofer_single_calc}
    \dFzero=r,\ \mathrm{s.t.}\ \Delta\theta_2\bigg|_{ d'=D/2}=\pi/8.
\end{align}
From \eqref{eq:deltatheta}, the {\it Fraunhofer distance} $\dFzero$ is calculated as
\begin{align}
\label{eq:Fraunhofer_single}
    \dFzero=2D^2 \sin^2(\theta)/\lambda
\end{align}
and thus, the {\it maximum Fraunhofer distance} denoted by $\overline{d}^{\mathrm{F0}}$ is obtained on the boresight ($\theta=\pi/2$) as 
\begin{align}
    \label{eq:Fraunhofer_single_max}
    \dFzeroMax={2D^2}/{\lambda}
\end{align}
\begin{remark}
    As seen in \eqref{eq:Fraunhofer_single} and Fig. \ref{fig:Fraunhofer_single}-b, the Fraunhofer distance $\dFzero$ which stands for the boundary between near-field and far-field regions is dependent on $\theta$. Therefore, as opposed to most works where the {\it Fraunhofer distance} is considered to be equal to $2D^2/\lambda$ which is independent of $\theta$,  in this work we name this as {\it maximum Fraunhofer distance}, and the angular dependant term expressed in \eqref{eq:Fraunhofer_single} as the {\it Fraunhofer distance}. %The same issue holds for the Fresnel and maximum Fresnel distances defined and obtained in the following.  
\end{remark}

In a similar way, if we consider $|\Delta \theta_3|\leq \pi/8$, the fourth term in \eqref{eq:deltatheta} can be neglected, and thus the {\it Fresnel distance function} denoted by $\dNzero$ is obtained by solving the following equation. 
\begin{align}
    \dNzero=r,\ \mathrm{s.t.}\ \Delta\theta_3\bigg|_{ d'=D/2}=\pi/8.
\end{align}
which results in $
    \dNzero= \sqrt{ |\cos \theta|\sin^2 \theta  {D^3}/{\lambda}}$.
The {\it maximum Fresnel distance} or simply {\it Fresnel distance}  $\overline{d}^{\mathrm{N0}}$ is then obtained by solving $\partial \dNzero/\partial \theta=0$, which yields $\theta=\tan^{-1}(\pm \sqrt{2})$ and 
\begin{align}
    \dNzeroMax=0.62\sqrt{{D^3}/{\lambda}}
\end{align}

\begin{figure}
    \centering
    \includegraphics[width=254pt]{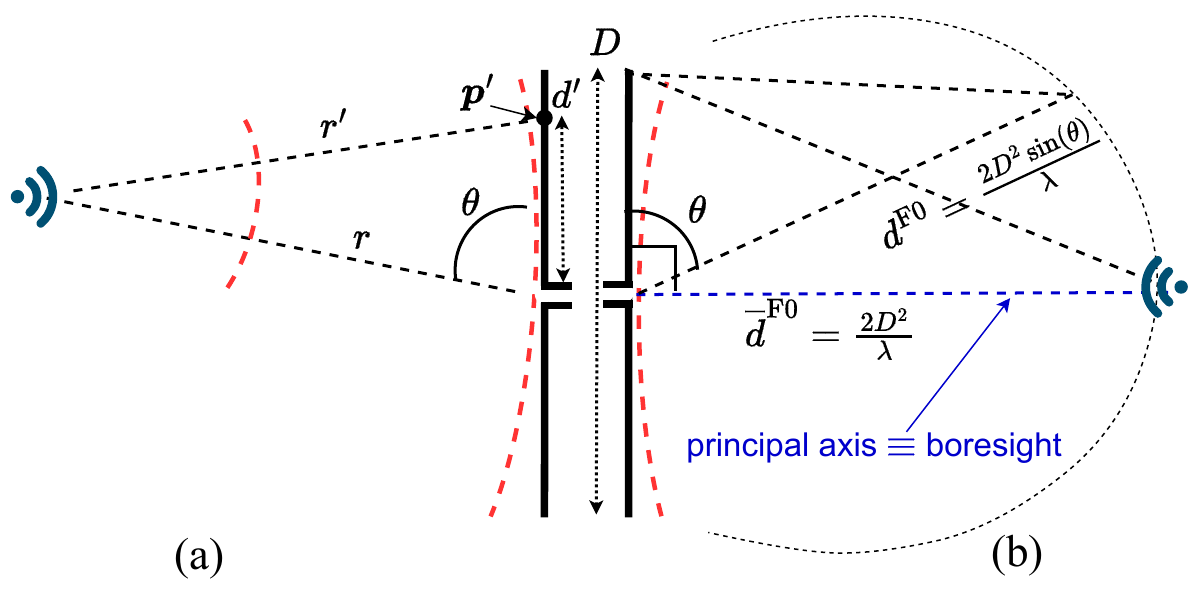}
    \caption{Antenna with diameter $D$. (a): UE located at an arbitrary distance.  (b): UE located at Fraunhofer distance.}
    \label{fig:Fraunhofer_single}
\end{figure}

% Now we consider the phased array.

% \begin{align}
%     %r_1=\sqrt{(x-x')^2+(y-y')^2+(z-z')^2}\big|_{(x',y',z')
%     %=(0,0,- \widehat{D}/2)}=
%     r_1=&\sqrt{r^2+r\widehat{D}\cos \theta +\widehat{D}{}^2/4 }
%     \\
%     r_N=&\sqrt{r^2-r\widehat{D}\cos \theta +\widehat{D}{}^2/4 }
% \end{align}

% \begin{multline}
%     \Delta\theta=\frac{2\pi}{\lambda}\left(r_N-r_1\right)
% =
%     \underbrace{\frac{-2\pi}{\lambda}\cos \theta \Delta D }_{\Delta\theta_1}
% +
% \underbrace{\frac{\pi}{\lambda}  \sin^2 \theta  (\Delta D)^2 r^{-1}}_{\Delta\theta_2^N}
% +
% \\
% \underbrace{\frac{\pi}{\lambda} \cos \theta\sin^2 \theta (\Delta D)^3 r^{-2}}_{\Delta\theta_3}
%     + \ ...
% \end{multline}

\subsection{Fraunhofer and Fresnel regions for phased array antennas}

% \begin{multline}
%     %\Delta\theta_{n,m}=
%     \angle a_n(\rbold) - \angle a_m(\rbold)=
%     \frac{2\pi}{\lambda}\left(r_n-r_m\right)
% =
%     \underbrace{\frac{-2\pi}{\lambda}|\cos \theta| (z_n-z_m) }_{\Delta\theta_{n,m}^{(1)}}
% +
% \\
% \underbrace{\frac{\pi}{\lambda}  \sin^2 \theta  (z_n-z_m)^2 r^{-1}}_{\Delta\theta_{n,m}^{(2)}}
% +
% \underbrace{\frac{\pi}{\lambda} |\cos \theta|\sin^2 \theta (z_n-z_m)^3 r^{-2}}_{\Delta\theta_{n,m}^{(3)}}
%     + \ ...
% \end{multline}

%review
In this part, we elaborate on how the results of Fraunhofer and Fresnel distances should be revisited to apply to phased array antennas. %As stated before, the Fraunhofer and Fresnel distances are characterized by studying the wavefront's curvature through analyzing the array elements' arrival phase. Therefore, the USW channel model applies here.  
 Consider a phased array antenna with the largest dimension of $D$. For simplicity we consider the uniform linear array (ULA) case, however, the results can be applied to the planar arrays as well. We consider $\frac{D}{\lambda}\geq 0.5$ which is commonly the case by having an inter-element spacing of half a wavelength and $N \geq 2$, or inter-element spacing lower than 0.5 (e.g., as in holographic surfaces) and a higher number of antenna elements.  %Assuming that the curvature of the wavefront is rather small on each antenna element , 
%we are interested in devising how distant the reference source point can be from the aperture in order to consider a planner  wavefront  in the whole aperture. 
As seen in Fig. \ref{fig:Fraunhofer_single}, for a single-element antenna of diameter $D$, the Fraunhofer distance is based on the maximum phase error corresponding to the maximum distance of $\pbold'$ to the center of the antenna which is equal to $D/2$. This is not however the case in phased array antennas. 

A phased array antenna consists of multiple elements, each having an independent feed. Unlike the conventional single-element antenna model, where the maximum curvature of the signal is measured relative to the central feed, the maximum phase delay in a phased array antenna is determined by the largest difference between any two elements of the array. To elaborate, considering a source point transmitting a spherical wave, depending on the distance from the source, each element of the phased array antenna receives a sinusoidal signal with a specific phase delay at its feed point. If the phase shift throughout different points of each antenna element is negligible (i.e., all array elements are in the far-field of the source point), the signal at each feed point is minimally distorted. Otherwise, the signal may exhibit some harmonic distortion. In either case, we can use the phase of the dominant harmonic of each feed point signal and compare it with the phase of any other feed point signal to assess the planarity of the wavefront. Thus, any feed point can serve as a reference for the phase delay comparison.
\begin{figure}
    \centering
    \includegraphics[width=254pt]{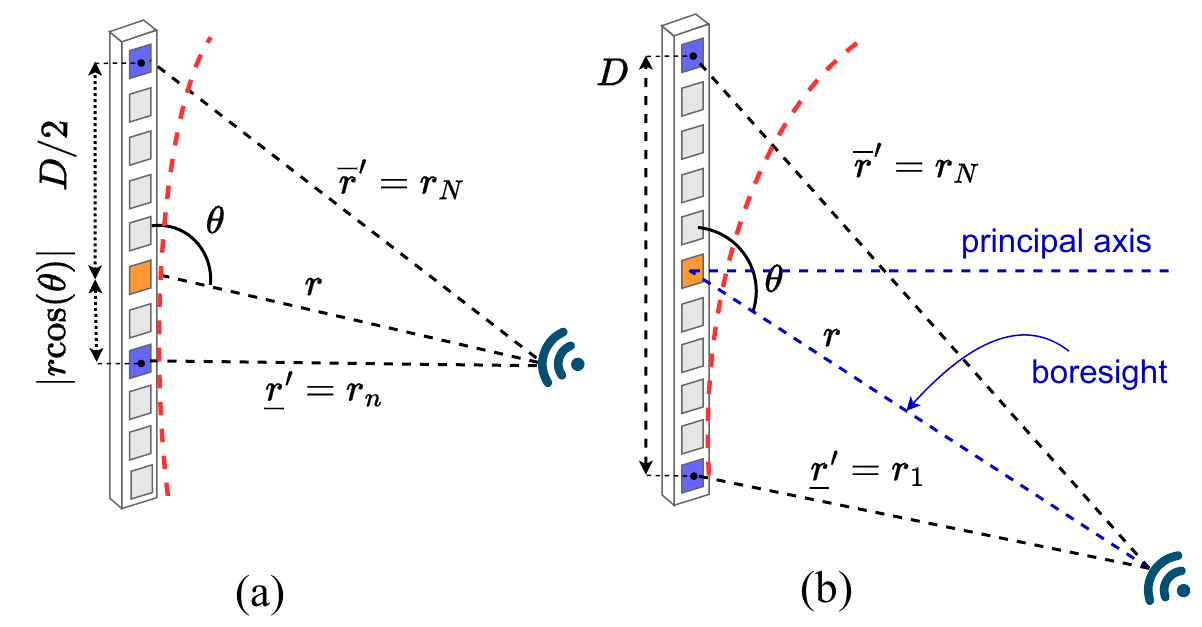}
    \caption{$N$-element ULA antenna with diameter $D$ for two scenarios regarding the position of the transmitter. 
}
    \label{fig:Fraunhofer_revisit}
\end{figure}
Fig. \ref{fig:Fraunhofer_revisit} illustrates an $N$-element ULA antenna with diameter $D$ representing two different cases for calculating the Fraunhofer distance. It is seen that the maximum phase delay is experienced for array elements with distance $\frac{D}{2}+|r\cos(\theta)|$ and $D$ corresponding to Figs. \ref{fig:Fraunhofer_revisit}-a and \ref{fig:Fraunhofer_revisit}-b respectively.\footnote{For the case shown in Fig. \ref{fig:Fraunhofer_revisit}-b, we have considered the approximation that the line perpendicular to the antenna plane (i.e., the line with distance $\underline{r}'$) coincides the center of some antenna element $n$.}
   Considering the two scenarios, the Fraunhofer distance is calculated as follows.
\begin{align}
\label{eq:Fraunhofer_array_calc22}
    \dF=r,\ \mathrm{s.t.}\ \Delta\theta_2\bigg|_{d'=D/2 + \Delta D}=\pi/8.
\end{align}
where
\begin{align}
\label{eq:DeltaD}
    \Delta D=\min \{ |r\cos \theta |,D/2 \}.
\end{align}
From \eqref{eq:Fraunhofer_array_calc22} and \eqref{eq:DeltaD}, $\dF$ is found by solving the following equation.
\begin{align}
\label{eq:dF110}
    \frac{2D^2 }{\lambda}\sin^2(\theta)\left( 1+\min\left\{ 1, \frac{2\dF|\cos{\theta}|}{D} \right\}
    \right)^2=\dF
\end{align}
If we consider $2\dF|\cos(\theta)|\leq D$, we can see that \eqref{eq:dF110} turns into a square polynomial equation. Let $\widetilde{F}(\theta)$ be the solution of this equation. In this case we have $\dF=\widetilde{F}(\theta)$. On the other hand, assuming $2\dF|\cos(\theta)|\leq D$ leads to the solution of \eqref{eq:dF110} as follows.
\begin{align}
    \label{eq:dF2}
    \dF=8D^2\sin^2(\theta)/\lambda
\end{align}
Besides, from \eqref{eq:dF110}, the Fraunhofer distance switches from $\widetilde{F}(\theta)$ to \eqref{eq:dF2}  at the angle the angle $\theta=\pi/2 - \theta^{\mathrm{F}}$ wherein the following equality holds.
\begin{align}
    \frac{8}{\lambda}D^2\sin^2(\theta)=\frac{D}{2|\cos(\theta)|}, \ \mathrm{for}\ \theta=\pi/2 - \theta^{\mathrm{F}}
\end{align}
where $\theta^{\mathrm{F}}=\pi/2-F^{-1}(\frac{\lambda}{2D})$, and $\theta^{\mathrm{F}}=\pi/2-F^{-1}(\frac{\lambda}{2D})$ in which $F^{-1}(\frac{\lambda}{2D})$ is the value of $\theta$ corresponding to the solution of $F(\theta)=\frac{\lambda}{2D}$ closet to $\pi/2$.
Considering this, together with \eqref{eq:dF2}, the Fraunhofer distance $\dF$ is obtained as follows.
\begin{align} 
\label{eq:dF234}
\dF =
\begin{cases}
     \widetilde{F}(\theta),
     & \text{if} \ \frac{\pi}{2}-\theta^{\mathrm{F}} \leq \theta \leq \frac{\pi}{2} +\theta^{\mathrm{F}}
 \\
8D^2\sin^2(\theta)/\lambda
    , & \text{otherwise } 
\end{cases} 
\end{align}

    As seen in \eqref{eq:dF234}, for $|\pi/2 - \theta|\geq \theta^{\mathrm{F}}$, we have $d^F=8D^2\sin^2 \theta/\lambda$. We call $\theta^{\mathrm{F}}$ the {\it Fraunhofer array angle}.
   The Fraunhofer array angle is tightly approximated as
    \begin{align}
        \label{eq:tFApprox}
        \theta^{\mathrm{F}}\approx \frac{1}{2}\sin^{-1}\left({\frac{\lambda}{8{D}}} \right)
    \end{align}
   
 The validity of the approximation can be easily verified by considering $\theta^{\mathrm{F}}=\pi/2 - 
 F^{-1}(\frac{\lambda}{2D})\approx \pi/2-F_0^{-1}(\frac{\lambda}{2D})$ for $D/\lambda \geq 0.5$ where
    \begin{multline}
        F(\theta)=8|\cos(\theta)|\sin^2(\theta) 
        = 
        8\cos(\theta)\sin(\theta)
        = {4}\sin(2\theta)
        \\
        \equiv F_0(\theta),
        \forall \theta \in [\underbrace{{\pi}/{2}- F^{-1}(1)}_{82.7^\circ},{\pi}/{2})
    \end{multline}
    
    % Fig .\ref{fig:thetaF} depicts the exact value of $\theta^{\mathrm{F}}=\pi/2-F^{-1}(\lambda/2D)$ as well as the approximated value from \eqref{eq:tFApprox} with various numbers of array elements $N$ corresponding to different values of $\frac{D}{\lambda}$ for a ULA antenna with half-wavelength inter-element spacing.
    
%    \begin{figure}
%     \centering
%     \includegraphics[width=254pt]{pictures/Simulations/thetaF.pdf}
%     \caption{The exact and approximated value of $\theta^{\mathrm{F}}$ for various number of array elements $N$ for a ULA antenna with half-wavelength inter-element spacing.}
%     \label{fig:thetaF}
% \end{figure}

\begin{lemma}
   The maximum Fraunhofer distance is 
   \begin{align}
\label{eq:dFbar}
    \dFMax
    &=
    8D^2\cos^2(\theta^{\mathrm{F}})/\lambda\approx 4 \times \dFzeroMax
\end{align}
where $\dFzeroMax$ is the maximum Fraunhofer distance of a single-element antenna having the same diameter as the intended phased array.
\end{lemma}
\begin{proof}
From \eqref{eq:dF234} we have $\dF(\theta)=\dF(\pi-\theta)$, and thus we only consider $\theta\in[0,\pi/2]$. Note that the term $\dF=8D^2\sin^2(\theta)/\lambda$ in \eqref{eq:dF234} is an increasing function of $\theta$ for $\theta\in(0,\pi/2-\theta^{\mathrm{F}}]$, and thus, we only need to show that $\dF<8D^2\sin^2(\theta)/\lambda\big|_{\theta=\pi/2-\theta^{\mathrm{F}}}= 8D^2\cos^2(\theta^{\mathrm{F}})/\lambda$  for $\theta\in[ \pi/2- \theta^{\mathrm{F}},\pi/2)$. It can be  verified that $\widehat{F}(\theta)$ is a monotonically decreasing function of $\theta$ for $F^{-1}(1)\leq F^{-1}(\frac{\lambda}{2D}) \leq \theta \leq \pi/2$ and thus the maximum Fraunhofer distance is obtained as \eqref{eq:dFbar}.
\end{proof}

The Fraunhofer distance versus $\theta$ for a single-element antenna (i.e., $N=1$) as well as ULAs with $N\in \{3,10,40\}$ is depicted in Fig. \ref{fig:Fraunhofer_array}. It is seen as $N$ increases, the value of $\theta^{\mathrm{F}}$ rapidly decreases toward zero, and for all values of $\theta$ other than the very small region around the boresight in the angular domain corresponding to $\theta\in[0,\pi/2-\theta^{\mathrm{F}}] \cup [\pi/2+\theta^{\mathrm{F}},\pi] $, the Fraunhofer distance is $8D^2\sin^2(\theta)/\lambda$ which is 4 times that corresponding to a single element antenna with the same diameter $D$.

\begin{figure}
    \centering
    \includegraphics[width=254pt]{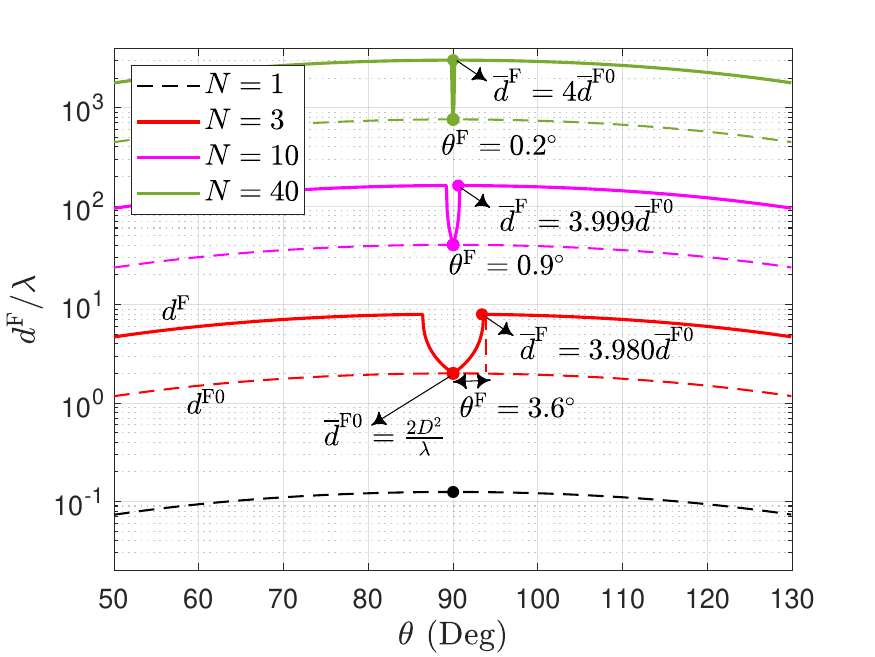}
    \caption{Fraunhofer distance per wavelength ($d^{\mathrm{F}}/\lambda$)  versus observation angle $\theta$ for $N$-element ULA with half-wavelength inter-element spacing (i.e., $D=\frac{(N-1) \lambda}{2}$). The dashed lines relate the Fraunhofer distance $d^{\mathrm{F0}}$ for the case of a single-element antenna whose diameter is the same as the corresponding antenna array diameter where $d^{\mathrm{F0}}=2D^2\sin^2(\theta)/\lambda$.%, and $\overline{d}^{\mathrm{F0}}=2D^2/\lambda$ is the maximum value of $d^{\mathrm{F0}}$.
    }
    \label{fig:Fraunhofer_array}
\end{figure}

Now we deal with the characterization of the Fresnel distance.
Similar to the discussion presented for calculating the Fraunhofer distance function for the two scenarios of Figs. \ref{fig:Fraunhofer_revisit}-a and \ref{fig:Fraunhofer_revisit}-b, the Fresnel distance is calculated from the following equation:
\begin{align}
\label{eq:Fraunhofer_array_calc}
    \dN=r,\ \mathrm{s.t.}\ \Delta\theta_3\bigg|_{d'=D/2 + \Delta D}=\pi/8
\end{align}
where $\Delta D$ is given in \eqref{eq:DeltaD}. From \eqref{eq:deltatheta}, \eqref{eq:DeltaD}, and \eqref{eq:Fraunhofer_array_calc},  the following equality is employed for deriving $\dN$. 
\begin{align}
\label{eq:dF11}
    \frac{D^3 }{\lambda}|\cos(\theta)|\sin^2(\theta)\left( 1+\min\left\{ 1, \frac{2\dN|\cos{\theta}|}{D} \right\}
    \right)^3=(\dN)^2
\end{align}
From \eqref{eq:dF11}, if $2\dN|\cos(\theta)|<D$, then $\dN$ is obtained by solving $y$ in the following qubic polynomial:
  \begin{align}
    \label{eq:ydN}
     \frac{1}{\lambda}|\cos(\theta)|\sin^2(\theta)\left(D+2y|\cos(\theta)|\right)^3-y^2=0
\end{align}
Otherwise, considering $2\dN|\cos(\theta)|\geq D$, from \eqref{eq:dF11} we have $\dN=\sqrt{ 8|\cos \theta|\sin^2 \theta  \frac{D^3}{\lambda}}$.  
Besides, from \eqref{eq:dF11}, it is seen that the angles $\theta_i^\mathrm{N}$ and $\widetilde{\theta}_i^\mathrm{N}$, $i\in\{1,2\}$ wherein the Fresnel distance switches from the first to second relations are obtained from the following equation:
\begin{align}
    \sqrt{\frac{8D^3}{\lambda} |\cos(\theta)|\sin^2(\theta)}=\frac{D}{2|\cos(\theta)|}
\end{align}
%which results in obtaining $\theta_i^{\mathrm{N}}$ and $\widetilde{\theta}_i^{\mathrm{N}}$
which corresponds to the solutions of $G(\theta)=16|\cos^3(\theta)|\sin^2(\theta)=\frac{\lambda}{2D}$. Therefore, the Fresnel distance function is obtained as  \begin{align} 
\label{eq:dNN}
\dN =
\begin{cases}
     \sqrt{ 8|\cos \theta|\sin^2 \theta  \frac{D^3}{\lambda}},
     & \text{if} \ \theta \in[\theta^\mathrm{N}_1,\theta^\mathrm{N}_2]\cup [\widetilde{\theta}^\mathrm{N}_1,\widetilde{\theta}^\mathrm{N}_2] 
 \\
y(\theta), & \mathrm{otherwise} %\mathrm{if}\ \theta \in [0,\pi]\setminus [\theta_1^N, \theta_2^N] \cup[\widetilde{\theta}_1,\widetilde{\theta}_2]  
\end{cases} 
\end{align}
where $y(\theta)\equiv y$ is obtained by finding the root of the  cubic polynomial equation \eqref{eq:ydN} in the corresponding region.
% \begin{align}
% \label{eq:453}
%     \frac{D}{\lambda} |\cos^3(\theta)|\sin^2(\theta)=\frac{1}{32}
% \end{align}
% Therefore    
\begin{lemma}
    The Fresnel distance $\dNMax$ for phased array antennas is obtained as 
    \begin{align}
\label{eq:dNNmax}
    \dNMax=1.75 \sqrt{\frac{D^3}{\lambda}}=\sqrt{8}\times \dNzeroMax
\end{align}
where $\dNzeroMax=0.62\sqrt{{D^3}/{\lambda}}$ is the Fresnel distance of a single-element antenna having the same diameter as that of the phased array.
\end{lemma}
\begin{proof}
Noting $\dN(\theta)=\dN(\pi-\theta)$ as seen in \eqref{eq:dF11}, we only consider  $\theta\in[0,\pi/2]$. First, we consider the lower and upper bounds for $\theta_1^{\mathrm{N}}$ and $\theta_2^{\mathrm{N}}$ respectively. Assuming the minimum value of ${D}/{\lambda}$ equal to 0.5, $\theta_1^\mathrm{N}$ and $\theta_2^\mathrm{N}$ are obtained from the solutions of $G^{-1}(1)$ which equal $15.1^\circ$ and $65^\circ$. It can be easily verified that for $D/\lambda$ higher than 0.5, %(corresponding to $N>2$ for arrays with inter-element spacing of half a wavelength),
we have $\theta_1^\mathrm{N}<15.1^\circ$ and $\theta_2^\mathrm{N}>65^\circ$. On the other hand, the maximum value of $\dN$ corresponding to $\theta\in[\theta_1^\mathrm{N},\theta_2^\mathrm{N}]$ is obtained for any value of $D/\lambda \geq 0.5$ at  $\theta^*=\tan^{-1}(\sqrt{2})=54.7^\circ$, by considering that $\theta_1^\mathrm{N}\leq\theta^*\leq \theta_2^\mathrm{N}$, $\partial \cos(\theta)\sin^2(\theta)/\partial \theta|_{\theta=\theta^*}=0$ and $\partial^2 \cos(\theta)\sin^2(\theta)/\partial \theta^2|_{\theta=\theta^*}<0$.
Noting that $\theta_1^{\mathrm{N}}$ and $\theta_2^{\mathrm{N}}$ are  boundary values for which $\dN$ switches between $ \sqrt{ 8|\cos \theta|\sin^2 \theta  \frac{D^3}{\lambda}}$ and $y(\theta)$. By verifying $y(\theta)<y(\theta_1^{\mathrm{N}})$ for $\theta<\theta_1^{\mathrm{N}}$, and $y(\theta)<y(\theta_2^{\mathrm{N}})$ for $\theta>\theta_2^{\mathrm{N}}$, we conclude that for any value of $D/\lambda\geq 0.5$, the global maximum of \eqref{eq:dNN} is realized at $\theta=\theta^*$ which is equal to \eqref{eq:dNNmax}.
\end{proof}
\begin{figure}
    \centering
    \includegraphics[width=254pt]{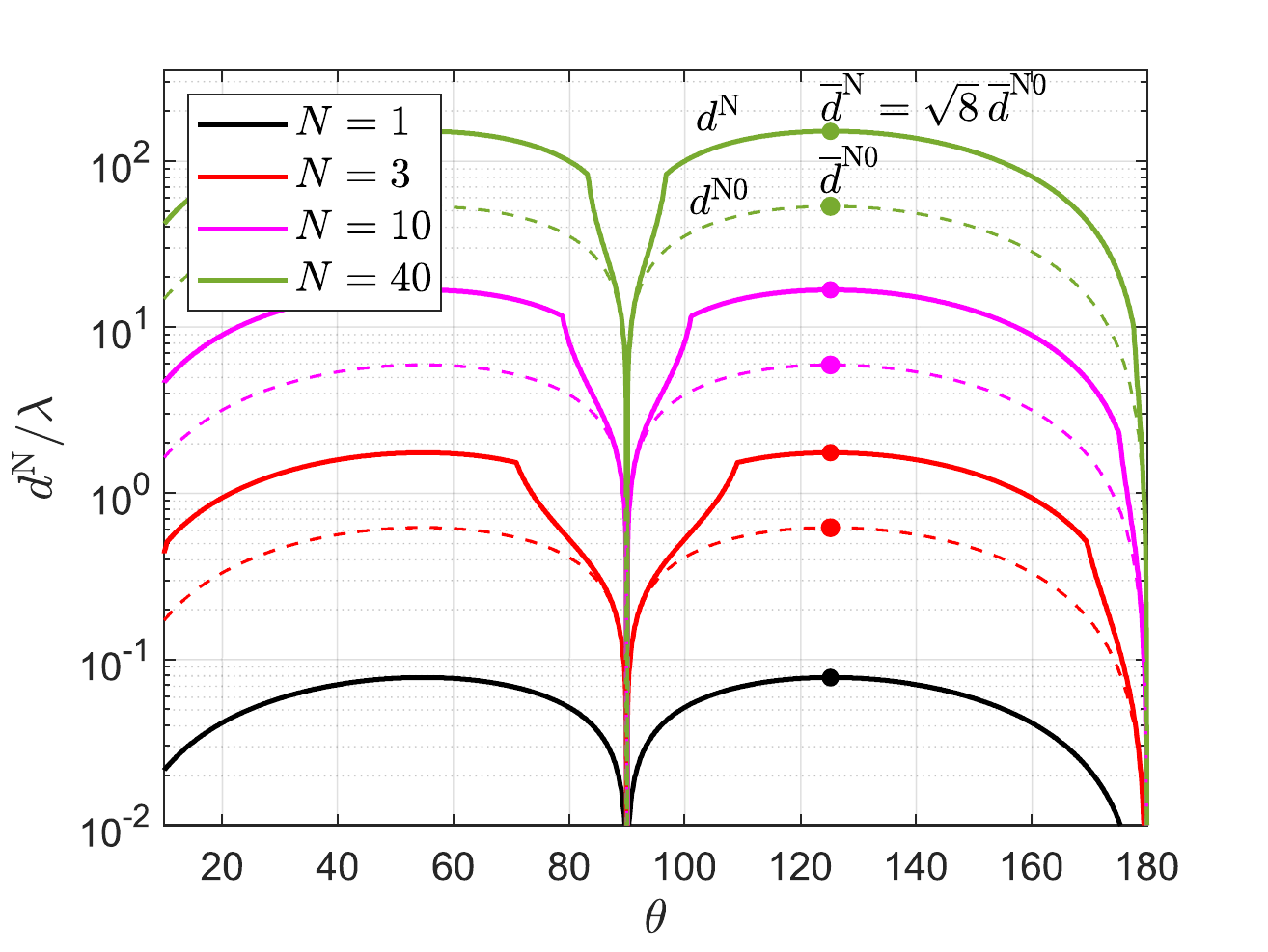}
    \caption{Fresnel distance function per wavelength ($\dN/\lambda$)  versus observation angle $\theta$ for $N$-element ULA with half-wavelength inter-element spacing. 
    The dashed lines relate the Fresnel distance $d^{\mathrm{N0}}$ for the case of a single-element antenna having an equal diameter to that of the corresponding array antenna.% where $d^{\mathrm{N0}}=2D^2\sin^2(\theta)/\lambda$, and $\overline{d}^{\mathrm{N0}}=0.62\sqrt{D^3/\lambda}$ is the maximum value of $d^{\mathrm{F0}}$.
    }
    \label{fig:Fresnel_array}
\end{figure}

The characterized Fresnel distance functions and the corresponding maximum values are depicted for a single-element antenna (i.e., $N=1$) as well as ULAs with $N\in \{3,10,40\}$ is depicted in Fig. \ref{fig:Fresnel_array}. The figure illustrates the relationship between the Fresnel distance functions for phased-array and single-element antennas having equal diameters, represented by solid and dashed lines, respectively. 

\section{Conclusions}
\label{sec:conclusions} 
In this paper, we revisited the characterization of the Fraunhofer and Fresnel regions established for single-element antennas, and provided appropriate expressions applicable to phased arrays.  We analytically derived corresponding closed-form solutions and showed that the Fraunhofer and Fresnel distances of phased arrays are increased compared to single-element antennas with the same dimensions.

\bibliographystyle{IEEEtran}

\bibliography{Mybib}

\end{document}